\documentclass{amsproc}
\usepackage{amsmath,amssymb,mathrsfs,stmaryrd}
\usepackage[cmtip,all,2cell]{xy}
\UseTwocells

\newtheorem{theorem}{Theorem}[section]
\newtheorem{lemma}[theorem]{Lemma}
\newtheorem{proposition}[theorem]{Proposition}
\theoremstyle{definition}
\newtheorem{definition}[theorem]{Definition}
\newtheorem{example}[theorem]{Example}
\theoremstyle{remark}
\newtheorem{remark}[theorem]{Remark}
\numberwithin{equation}{section}

\newcommand\Rb{{\mathbb R}}  
\newcommand\Zb{{\mathbb Z}}  
\newcommand\Cb{{\mathbb C}}  
\newcommand\Lc{{\mathcal L}} 
\newcommand\M{{\mathcal M}}  
\newcommand\Sc{{\mathcal S}} 
\newcommand\Xs{{\mathscr X}} 
\newcommand\Ps{{\mathscr P}} 
\newcommand\Gs{{\mathscr G}} 
\newcommand\Gc{{\mathcal G}} 

\newcommand{\arrows}{\rightrightarrows} 

\newcommand{\wt}{\widetilde}
\newcommand{\To}{\longrightarrow}

\newcommand{\lTo}{\longleftarrow}
\newcommand{\dashto}{\dashrightarrow}

\newcommand{\into}{\hookrightarrow}
\newcommand{\tof}[1]{\stackrel{#1}{\rightarrow}}

\newcommand{\Tof}[1]{\stackrel{#1}{\longrightarrow}}

\newcommand{\wh}{\widehat}
\DeclareMathOperator{\Hom}{Hom}
\DeclareMathOperator{\AUT}{AUT}
\DeclareMathOperator{\Ch}{Ch}
\DeclareMathOperator{\id}{id}
\DeclareMathOperator{\Aut}{Aut}
\DeclareMathOperator{\HOM}{HOM}

\DeclareMathOperator{\ev}{ev}
\DeclareMathOperator{\Aff}{Aff}
\DeclareMathOperator{\LCH}{\mathcal{LCH}}
\DeclareMathOperator{\Pic}{\textbf{Pic}} 
\DeclareMathOperator{\chain}{\textbf{Ch}^{[-1,0]}} 

\DeclareMathOperator{\Pont}{{Pont}} 

\begin{document}

\title{Geometric T-dualization}
\author{Calder Daenzer}
\address{Department of Mathematics, Pennsylvania State University, University Park, Pennsylvania 16801}
\email{daenzer@math.psu.edu}
\date{September 30, 2013}
\maketitle
\begin{abstract} In this article we realize T-duality as a geometric transform of bundles of abelian group stacks.  The transform applies in the algebro-geometric setting as well as the topological setting, and thus makes precise the link between the models of T-duality in these two settings, which are superficially quite different.  The transform is also valid for torus bundles with affine structure group, and thus provides a constructive method of producing T-duals in the affine case.
\end{abstract}
\tableofcontents
\section{Introduction}
After the discovery of T-duality as an equivalence relating different string theory types \cite{Bus}, there arrived numerous mathematical models of the phenomenon. The result has been the growth of a strikingly wide range of mathematical descriptions of T-duality.  For instance:
\begin{itemize}
\item In algebraic geometry, T-duality it is typically viewed as a gerby enhancement of a relative Jacobian operation for toric fibrations.  In appropriate cases the duality induces a Fourier--Mukai equivalence between the derived categories of the dual fibrations (see \cite{DonPan}).
    \item In topology and noncommutative geometry, T-duality is defined as a relation between a $U(1)$-gerbe on a principal torus bundle and a ``dual'' gerbe on a ``dual'' principal torus bundle.  The dual bundle may be a family on noncommutative tori.  One consequence of the relation is that the dual objects have isomorphic K-theory (see \cite{MatRos}, \cite{Dae}, \cite{BunSch}).
\item In differential geometry, T-duality is defined as a relation on the set of pairs $\{P,H\}$ in which $P$ is a principal torus bundle over a fixed base manifold, and $H\in\Omega^3(P)$ is a $T$-invariant form.  This relation induces an isomorphism from the (twisted) generalized geometrical structures on $P$ to those on the T-dual bundle, which changes the structure type (for instance a complex structure on $P$ can map to a symplectic structure on the dual, see \cite{CavGua}).
\end{itemize}

There has always been a tacit assumption that these models are indeed the different mathematical faces of a single phenomenon, but only in specific cases is this easy to verify.
In fact, only in the case of a single torus $T=V/\Lambda\simeq \Rb^n/\Zb^n$ is it truly obvious that the T-dual (in every model) should be given by
\[ T^\vee:=\Hom(\Lambda,U(1)), \]
that is, the Pontryagin dual of the fundamental group $\Lambda=\pi_1(T)$.

The purpose of this article is to describe a transform (\textbf{T-dualization}), which may be applied equally well in the topological, differential geometric, or algebro-geometric categories, and which should reproduce in each setting the correct notion of T-duality.

Beyond the unifying effect of producing such a transform, we will see two concrete results.  For one, it allows T-duality to be recast as a functorial operation rather than a many-to-many relation.  Secondly, the transform is not only valid for principal bundles with gerbe, but also for affine torus bundles with gerbe.  Thus it provides a constructive method of producing the T-duals to affine bundles, which were originally calculated using cohomological data by Baraglia \cite{Bar}.  (The affine case is important because in the SYZ-picture of mirror symmetry (\cite{SYZ}), the nonsingular loci of the toric fibrations in a mirror pair of Calabi--Yau 3-folds admit an affine structure, but typically not the structure of a principal bundle.)

Just to provide a rough picture of what T-dualization is, let us examine the case of a single torus $T=V/\Lambda$.  Dualization is defined in two steps:
\begin{enumerate}
\item Pass from $T$ to the stacky quotient $[T//V]$ of $T$ by $V$.
\item Apply a stacky version of Pontryagin duality\footnote{Pontryagin duality (or Cartier duality) for abelian group stacks has been used in \cite{DonPan} (Appendix A) and in \cite{BSST}.  Our definition does not coincide with the one which appears in those works, see Remark \eqref{R:Other Pont Duality defs}.} to $[T//V]$.
\end{enumerate}
Pontryagin duality for abelian group stacks will be defined in Section \eqref{S:Pontryagin Duality}, and it will be shown that $\Pont[T//V]$ is equivalent to the group $T^\vee=\Hom(\Lambda,U(1))$ (and thus agrees with the standard definition of the T-dual torus).  For now, this fact should at least seem plausible, in view of the isomorphism $[V//\Lambda]\simeq[*//\Lambda]$.

Thus in one line, $T$-dualization of a single torus is the procedure
 \[ T\longrightarrow [T//V] \rightsquigarrow \Pont([T//V]). \]
The general method of T-dualization will then amount to finding a suitable notion of modding out $V$ and forming the Pontryagin dual in the context of bundles of abelian group stacks.

\begin{remark}At first glance, it does not seem that T-dualization for abelian group stacks could apply to principal bundles with gerbe, because a principal torus bundle is not a bundle of groups (or group stacks).  But in fact, any abelian principal bundle can be canonically embedded into a bundle of groups (the bundle is the disjoint union of the $n$-th powers of the principal bundle, for all $n\in\Zb$), and a similar embedding applies for affine torus bundles.
\end{remark}
\noindent\textbf{Outline. } In Section \eqref{S:Short Complexes} we recall chain complex presentations of abelian group stacks.  In Section \eqref{S:Pontryagin Duality} we define Pontryagin duality using the chain complex presentations, and show that the result is independent of choice of presentation, and thus defines a duality for (representable) group stacks.  In Section \eqref{S:Picard Bundles} we upgrade Pontryagin duality to bundles of abelian group stacks, and describe the appropriate bundles of chain complexes which are presentations of these stacks.  In Section \eqref{S:Construction of T-duals} we define T-dualization, and finally, in Section \eqref{S:Relation} we describe how our construction produces the other models of T-duality.

\noindent\textbf{Acknowledgements. } The author would like to thank Martin Olsson, Ping Xu, and Peter Dalakov for many helpful discussions and insights.

\section{Presentations of Picard stacks}\label{S:Short Complexes} In this section we will recall two equivalent types of presentations of abelian group stacks (called Picard stacks).  One type of presentation is via Picard groupoids (also called abelian 2-groups).  The other type of presentation, due to Deligne, is via 2-term chain complexes of groups (which could also be called abelian crossed modules).  We will then examine morphisms of representable Picard stacks, and derive ways of expressing them in terms of exact sequences of groups.

\begin{definition} A \textbf{Picard groupoid} is a groupoid internal to the category of abelian groups.  We denote by \textbf{Pic} the strict 2-category of groupoids, functors, and natural transformations internal to abelian groups.  These functors and natural transformations will be referred to as \textbf{additive} functors and \textbf{additive} natural transformations.
\end{definition}
\begin{definition}
Let $\chain$ denote the following 2-category:
\begin{itemize}
\item Objects are two-term chain complexes of abelian groups $ [A\tof{d}B]$, with $A$ in degree $-1$ and $B$ in degree $0$. For convenience we refer to these as \textbf{short chain complexes}.
\item 1-Arrows are degree zero chain morphisms
 \[f=(f_A,f_B):[A\to B]\To[A'\to B']. \]
\item 2-Arrows are chain homotopies.  Thus for a pair of chain morphisms \linebreak$[A\tof{d} B]\tof{f,f'}[A'\tof{d'}B']$, a 2-arrow $\eta:f\Rightarrow f'$ is a degree $-1$ homomorphism satisfying $[d,\eta]:= d'\circ\eta+\eta\circ d = f'-f$.  This reduces to a homomorphism $\eta:B\to A'$ satisfying $d'\circ \eta=f'_B-f_B$ and $\eta\circ d= f'_A-f_A$.
\end{itemize}
\end{definition}
In what follows we will want to be comfortable using Picard groupoids and short chain complexes interchangeably as presentations of Picard stacks.  Thus let us carefully write down an equivalence of 2-categories between $\Pic$ and $\chain$.

\noindent\underline{Short complexes from Picard groupoids:}  From a Picard groupoid $G=(G_1\stackrel{s,t}{\arrows} G_0)$ one obtains a short chain complex
\[ \textbf{C}(G):=[\ker(s)\Tof{t|_{\ker s}}G_0]. \]
An additive functor $f:G\to G'$ gives rise to a chain morphism $\textbf{C}(f):\textbf{C}(G)\To \textbf{C}(G')$ whose components are $\textbf{C}(f):=(f|_{\ker(s)},f|_{G_0})$.
An additive natural transformation $\xymatrix{G\rtwocell^f_{f'}{_\eta} & G'}$
with underlying set function $\eta:G_0\to G'_1$ induces a function
\[ \textbf{C}(\eta):G_0\To \ker(s'),\quad b\longmapsto \eta(b)-s(\eta(b)) \]
which is easily verified to be a chain homotopy from $\textbf{C}(f)$ to $\textbf{C}(f')$.  Putting this all to together, we have a strict 2-functor
\[ \textbf{C}:\Pic\To \chain. \]

\noindent\underline{Picard groupoids from short complexes:} Given a complex $A\tof{d}B$, one has an action of $A$ on $B$ by $a\cdot b:=d(a)+b$.  Associate to this data the \textbf{action groupoid}, which is by definition
 \[ A\ltimes_d B= (A\times B\arrows B) \]
with source, range, and composition given by
\[ s(a,b)=b,\quad t(a,b)=da+b,\quad (a,d(a')+b)\circ (a',b):=(a+a',b). \]
The group structure on $A\times B$ (that is, $(a,b)+(a',b')=(a+a',b+b')$) is functorial, and makes $A\ltimes_d B$ into a Picard groupoid $\textbf{P}[A\tof{d} B]$.

A chain morphism $f:[A\to B]\To [A'\to B']$ induces an additive functor
\[ \textbf{P}(f):=f_A\times f_B:A\ltimes B\to A'\ltimes B'.\]
And finally a homotopy $\eta:f\Rightarrow f'$ between two such chain morphisms induces an additive natural transformation $\textbf{P}(\eta):\textbf{P}(f)\Rightarrow \textbf{P}(f')$ whose underlying function is
\[ (A\ltimes B)_0\to (A'\ltimes B')_1,\qquad b\longmapsto (\eta(b),f(b)). \]
(Note that $(\eta(b),f(b))$ is indeed an arrow from $f(b)$ to $f'(b)$.)
Putting this all together, we have a strict 2-functor $\textbf{P}:\chain\To \Pic$.
It is easy to see that $\textbf{C}(\textbf{P}([A\to B]))\equiv [A\to B]$.  Furthermore, for any Picard groupoid $G=(G_1\stackrel{s,t}{\arrows}G_0)$, we have an isomorphism
\[ G\Tof{\sim} \ker(s)\ltimes G_0\equiv \textbf{P}(\textbf{C}(G)),\qquad g\longmapsto (g-s(g),s(g)), \]
and this induces a pseudonatural isomorphism $\id_{\Pic}\Longrightarrow \textbf{P}\circ\textbf{C}$. This equivalence was pointed out in \cite{Del}, and we paraphrase it as:
\begin{proposition}\cite{Del} The strict 2-functors
\[ \textbf{P}:\chain\leftrightarrows \Pic:\textbf{C} \]
are equivalences of 2-categories, with $\textbf{C}\circ \textbf{P}=\id_{\chain}$ and $\id_{\Pic}\simeq \textbf{P}\circ\textbf{C}$.
\end{proposition}

This purely algebraic equivalence upgrades without change when topology is introduced.  Thus if $\Sc$ is a site and $\textbf{Ab}_\Sc$ is the category of abelian groups internal to $\Sc$, then we have the 2-category $\Pic_\Sc$
of groupoids, functors, and natural transformations internal to $\textbf{Ab}_\Sc$, and it is equivalent to the 2-category $\textbf{Ch}_\Sc^{[0,1]}$ of short chain complexes $[A\to B]$ of abelian groups in $\Sc$.

For the constructions of Pontryagin duality and T-duality which follow, we will for the most part use Picard groupoids and short complexes.  However, some proofs require the use of Picard stacks, so we will recall the necessary terminology here.

\begin{definition} A \textbf{Picard $\Sc$-stack} is a $\Pic$-valued $\Sc$-stack. The 2-category $\Sc-\textbf{PicStacks}$ of Picard $\Sc$-stacks is the 2-category of $\Pic$-valued $\Sc$-stacks. \end{definition}
An equivalent but less succinct definition is that a Picard stack is a groupoid-valued stack $\Xs$ together with a morphism $\Xs\times\Xs\tof{+} \Xs$ whose value at each point $T\in\Sc$:
\[ \Xs(T)\times\Xs(T)\tof{+}\Xs(T) \]
makes a $\Xs(T)$ a Picard groupoid.
\begin{example}If $A\ltimes B\in\Pic_\Sc$, then the 2-functor
\[ [B//A]:\Sc^{op}\To \Pic \]
which sends $T\in \Sc$ to the Picard groupoid of right principal $(A\ltimes B)$-bundles (or $A\ltimes B$-torsors) over $T$ is a Picard stack.  The group law of $[B//A]$ is of course the one induced by addition in $A\ltimes B$.
\end{example}
Propositions \eqref{P:P:PicardMoritaEquivs} and \eqref{P:Stack Morphisms are groups} below show that equivalences of Picard stacks may be phrased in terms of exact sequences of groups in $\Sc$.  This will be the crucial fact when extending Pontryagin duality to Picard stacks.
\begin{remark} A sequence of groups $0\to A\tof{\alpha} B\tof{\beta} C\to 0$ in a site $\Sc$ will be called \textbf{exact} only when the associated sequence of sheaves of groups is exact\footnote{We will assume all sites are subcanonical, and tacitly identify an object $X\in \Sc$ with the sheaf $\Hom_\Sc(-,X)$ that it determines.}.  This means that $\alpha$ is an isomorphism onto its image, and $\beta$ is an \textbf{epimorphism}, that is, $\beta$ admits local sections.  Also, a chain morphism $C^\bullet\to D^\bullet$ between chain complexes in $\Sc$ is called a \textbf{quasi-isomorphism} when the associated map of complexes of sheaves is a quasi-isomorphism.
\end{remark}

\begin{proposition}\label{P:P:PicardMoritaEquivs} Let $A\ltimes B\Tof{\phi}A'\ltimes B'$ be a morphism of Picard $\Sc$-groupoids.  Then the following are equivalent:
\begin{enumerate}
\item $\phi$ induces an isomorphism of Picard stacks.
\item $\phi$ is a Morita equivalence.
\item The sequence \[0\to A\tof{i} P_\phi\tof{\epsilon} B'\to 0\]
is exact, where $i(a)=(da,\phi(a),0)\in P_\phi=B\times_{\phi,B',t}(A\times B)$, and $\epsilon(a,a',b')=b'$.
\item The sequence
\[ 0\to A\tof{\alpha} B\times A'\tof{\beta} B'\to 0, \]
is exact, where $\alpha(a):=(da,\phi(a))$ and $\beta(b,a'):=\phi(b)-d'a'$.
\item The chain morphism $[A\to B]\Tof{\b{C}(\phi)}[A'\to B']$ is a quasi-isomorphism.

\end{enumerate}
\end{proposition}

\begin{proof}(1)$\Leftrightarrow$(2) can be found in any introduction to stacks (see e.g. \cite{Met}, \cite{BehXu}).  To obtain (2)$\Leftrightarrow$(3), note that $P_\phi$ is the right principal $(A\ltimes B)-(A'\ltimes B')$-bimodule associated to the functor $\phi$.  It is easy to see that the left groupoid action
\[ (A\ltimes B)\times_B P_\phi\To P_\phi \]
is principal as well precisely when (3) is satisfied.  And being left principal is the same as being a Morita equivalence.

The sequences in (3) and (4) are isomorphic, indeed
\[ P_\phi=B\times_{B'}A'\times B'\To B\times A',\qquad (b,a',b')\mapsto (b,a') \]
induces the isomorphism, and the inverse is $(b,a')\mapsto(b,a',\phi(b)-d'a')$.

The mapping cone complex associated to (5) is the sequence (4), thus sequence (4) is exact if and only if (5) is a quasi-isomorphism.
\end{proof}

It is proved in \cite{Del} that \emph{every} Picard $\Sc$-stack is presentable by a short chain complex $[A\to B]$, however, $A$ and $B$ are in general not groups in $\Sc$, but sheaves of groups on $\Sc$.  For our purposes it will be important to always have actual objects of $\Sc$ presenting stacks, and also to have objects of $\Sc$ presenting morphisms between stacks. In Proposition \eqref{P:Stack Morphisms are groups} below, we show that if two stacks are representable, then the morphisms between them may be represented by groups.

A Picard $\Sc$-stack $\Xs$ is called \textbf{representable} if there is a group $B\in \Sc$ and an additive representable epimorphism $B\to \Xs$, that is, if for every space $X\in \Sc$ and morphism $X\to \Xs$, the fiber product $X\times_\Xs B$ is (the sheaf associated to) an object of $\Sc$, and the pullback map $X\times_\Xs B\to X$ is an epimorphism.

Thus in particular, if $\Xs$ is representable then $G:=B\times_\Xs B$ is a space.  Furthermore, $G$ is endowed in a canonical way with the structure of a Picard groupoid in $\Sc$, whose associated stack is (isomorphic to) $\Xs$.  The groupoid $G$ or its associated short complex, will be referred to as a \textbf{presentation} of $\Xs$.

\begin{proposition}\label{P:Stack Morphisms are groups} Suppose that $\Xs\simeq[B//A]$ and $\Xs'\simeq[B'//A']$ are representable Picard $\Sc$-stacks, and $\Xs\tof{\phi} \Xs'$ is a homomorphism. Then
\begin{enumerate}
\item The right principal $(A\ltimes B)-(A'\ltimes B')$-bimodule $P:=B\times_{\{\phi\circ j,\Xs',j'\}} B'$ is an abelian group in $\Sc$.  Here $B\tof{j}\Xs$ and $B'\tof{j'}\Xs'$ are the canonical maps.
\item There are homomorphisms $A\to P$ and $A'\to P$ such that the resulting action groupoid $A\times A'\ltimes P$ fits into a zig-zag of functors
\[ (A\ltimes B)\stackrel{M. equiv.}{\longleftarrow}(A\times A'\ltimes P)\To (A'\ltimes B') \]
which induces the morphism $\Xs\tof{\phi}\Xs'$.  The leftward arrow is a Morita equivalence, and the rightward arrow is as well if and only if $\phi$ is an isomorphism.
\item The morphism $\Xs\tof{\phi}\Xs'$ is induced by a zig-zag of chain morphisms
\[ [A\to B]\stackrel{quasi-isom}{\longleftarrow} [A\times A'\to P]\To [A'\to B'] .\]
The leftward arrow is a quasi-isomorphism, and the rightward arrow is as well if and only if $\phi$ is an isomorphism.
\end{enumerate}
\end{proposition}
\begin{proof}
For any test space $T\in \Sc$, we have by definition of the fibered product,
\[ B\times_{\{\phi\circ j,\Xs',j'\}}B'(T)=B(T)\times_{\Xs(T)_0}\Xs(T)_1\times_{\Xs(T)_0}B'(T), \]
which is an abelian group since $B$, $\Xs$, and $B'$ are, and $\phi\circ j$ and $j'$ are homomorphisms. Thus (1) holds.

To prove (2), first note that the left groupoid action $(A\ltimes B)\times_B P\tof{\mu} P$ acting on the identity object $0_P\in P$ induces a homomorphism
\[ A\Tof{d} P,\qquad da:= \mu(a,0_P), \]
and similarly a homomorphism $A'\tof{d'} P$.  Thus we can form an action groupoid $A\times A'\ltimes P$, and the canonical maps $P\to B$ and $P\to B'$ induce obvious functors
\[  (A\ltimes B)\lTo(A\times A'\ltimes P)\To (A'\ltimes B'). \]
The leftward functor is a Morita equivalence precisely because the right action of $A'\times B'$ on $P$ is principal. At stack level, the left facing arrow becomes an isomorphism, and the right facing arrow becomes $\phi$.   The rightward functor is an equivalence if and only if the $(A\ltimes B)$-action on $P$ is principal, that is, if and only if $\phi$ is an isomorphism.  Thus (2) is proved.

Statement (3) is simply a restatement of (2) in terms of short chain complexes.
\end{proof}
\section{Pontryagin duality for Picard stacks}\label{S:Pontryagin Duality}

Recall that for a topological group $A$, the Pontryagin dual is the group of homomorphisms
\[ \wh{A}=\Hom(A,U(1)),\]
equipped with its compact-open topology.  It is a classical theorem that if $A$ is a locally compact Hausdorff abelian group, then so is the dual $\wh{A}$, and the canonical evaluation map gives an isomorphism:
\[ ev:A\Tof{\sim} \wh{\wh{A}},\qquad ev(a)(\phi):=\phi(a)\equiv\langle\phi,a\rangle. \]
Thus Pontryagin duality is a contravariant auto-equivalence of the category of locally compact Hausdorff (LCH) abelian groups.  Furthermore, if a sequence
\[ 0\to A\tof{\alpha} B\tof{\beta}C\to 0 \]
of LCH abelian groups is topologically exact (meaning it is algebraically exact, and $\alpha:A\to\alpha(A)$ is a homeomorphism, and $\beta$ is a topological quotient map), then so is the Pontryagin dual sequence\footnote{A sequence exhibiting the necessity of \emph{topological} exactness is
\[ 0\to \Zb\tof{\theta}\Rb/\Zb\to (\Rb/\Zb)/\theta(\Zb)\to 0, \]
where $\theta$ is multiplication by an irrational number.  Since the image of $\theta$ is dense, it cannot be an isomorphism onto its image.  This sequence is algebraically exact but its Pontryagin dual is not.}.

Now we proceed towards the extension of Pontryagin duality to Picard stacks.  For this we must define the correct site:
\begin{definition} Let $\LCH$ denote the site whose underlying category is LCH and whose covers are the surjective open maps. \end{definition}

As a first step, we extend Pontryagin duality to Picard groupoids:

\begin{proposition} Pontryagin duality extends to an arrow reversing automorphism of the 2-category $\Pic_{\LCH}$ of locally compact Hausdorff Picard groupoids.  The dual is given by:
\[ A\ltimes_d B\rightsquigarrow \wh{B}\ltimes_{\wh{d}}\wh{A}. \]
\end{proposition}
\begin{proof}
We only need to check that this operation is 2-functorial.  After that it is automatically an automorphism.  This takes little effort once the description of Picard groupoids by short complexes is in hand.  Indeed, to see that the Pontryagin dual of a morphism $A\ltimes B\tof{f}A'\ltimes B'$ is again an additive functor, simply look at the diagram:
\[\xymatrix{ A\ar[r]^d\ar[d]_{f_A} & B\ar[d]^{f_B} \\ A'\ar[r]^{d'}& B'}\qquad\rightsquigarrow\qquad
\xymatrix{ \wh{B}\ar[r]^{\wh{d}} & \wh{A} \\ \wh{B'}\ar[r]^{\wh{d'}}\ar[u]_{\wh{f_B}}& \wh{A'}\ar[u]_{\wh{f_A}}}. \]
Similarly, to see that additive transformations dualize to additive transformations, observe that a chain homotopy $B\to A'$ is sent to a chain homotopy $\wh{A'}\to\wh{B}$.
\end{proof}
Note that $\wh{B}\ltimes\wh{A}$ may be identified with the groupoid whose objects are the continuous additive morphisms from $(A\ltimes B)$ to $(U(1)\ltimes *)$, and whose arrows are the continuous additive natural transformations between such morphisms.

Now to see that Pontryagin duality for Picard groupoids descends to stacks, we will need to understand its effect on Morita equivalences.  And for this we need a quick lemma:
\begin{lemma} A sequence of topological groups
\[ 0\to A\tof{\alpha} B\tof{\beta}C\to 0 \]
is topologically exact if and only if it is exact as a sequence of sheaves in the site $\LCH$.
\end{lemma}
\begin{proof} It is immediate that $\alpha$ is an isomorphism to its image if and only if the map of sheaves is injective.  Next we must show that $\beta$ is a topological quotient map if and only if it induces an epimorphism of sheaves.  For this, first observe that $\beta$, being both a quotient map and a homomorphism, is itself open (indeed, if $U\subset B$ is open, then so is $\beta^{-1}\beta(U)=\cup_{a\in \ker\beta}Ua$, and then since $C$ has the quotient topology $\beta(U)$ is open as well).  Thus $\beta$ is itself a cover in $\Sc$, thus trivially induces an epimorphism of sheaves.  On the other hand, if $\beta$ induces an epimorphism of sheaves, then choose a cover $\wt{C}\tof{j}C$ over which $\beta$ admits a section $B\stackrel{s}{\leftarrow}\wt{C}$.  Then $j=\beta\circ s$, and $j$ is open so $\beta$ must be as well.  Thus $\beta$ is open and surjective, and thus in particular a topological quotient map.
\end{proof}

\begin{proposition}\label{P:Pont of Morita is Morita} A morphism $A\ltimes B\tof{\phi} A'\ltimes B'$ in $\Pic_{\LCH}$ is a Morita equivalence if and only if the dual morphism $\wh{B'}\ltimes\wh{A'}\tof{\phi^*}\wh{B}\ltimes\wh{A}$ is a Morita equivalence.
\end{proposition}
\begin{proof} According to Proposition \eqref{P:P:PicardMoritaEquivs}, $A\ltimes B\tof{\phi} A'\ltimes B'$ is a Morita equivalence if and only if a sequence of the form
\[ 0\to A\tof{\alpha} B\times A'\tof{\beta} B'\to 0 \]
is topologically exact. Dualizing this sequence results in another exact sequence which expresses the fact that $\wh{B'}\ltimes\wh{A'}\tof{\phi^*}\wh{B}\ltimes\wh{A}$ is a Morita equivalence.
\end{proof}

\begin{theorem} Pontryagin duality for Picard groupoids descends to representable Picard $\LCH$-stacks. \end{theorem}
\begin{proof} Our proposed definition of the Pontryagin dual of a representable Picard stack $\Xs\simeq[B//A]$ is of course
\[ \Pont(\Xs):=[\wh{A}//\wh{B}].\]
To make this operation functorial, one should first choose for each such Picard stack a preferred presentation $[A\to B]$.  (Or to avoid this choice, one could use the direct sum $\coprod B$ of all representable epimorphisms to $\Xs$ by groups $B$, which is again a locally compact group.)

Next, by Proposition \eqref{P:Stack Morphisms are groups} Part (3), a morphism $\Xs\tof{\phi}\Xs'$ of such Picard stacks may be presented by a zig-zag of short chain complexes of groups, and dualizing these morphisms provides the Pontryagin dual morphism $\Pont(\Xs')\tof{\phi^*}\Pont(\Xs)$.  In particular, after Proposition \eqref{P:Pont of Morita is Morita}, a morphism of Picard stacks is an isomorphism if and only if the Pontryagin dual morphism is.
\end{proof}
\begin{remark}\label{R:Other Pont Duality defs}
There is a definition of Pontryagin duality in the literature (see e.g. \cite{DonPan},\cite{BSST}), which is that the Pontryagin dual of $[B//A]$ is the group stack of homomorphisms to $[*//U(1)]$:
 \[ D([B//A]):=\HOM([B//A],[*//U(1)]). \]
It is shown in \cite{Del} that a short chain complex presentation of $D([B//A])$ is given by the truncation to degrees $\{-1,0\}$  of the derived functor $R^\bullet\Hom(A\to B,U(1)\to 0)$.  This dual does not in general agree with our dual, whose short complex presentation is the un-derived Hom.

For many groups the derived and un-derived duals are equivalent (see \cite{BSST}), but this can no longer be expected to hold once we pass to bundles of groups.  It seems that our definition results in the kernel of the canonical morphism \linebreak$\HOM([A//B],*//U(1))\to \HOM( [*//B],[*//U(1)]$.
\end{remark}

\section{Bundles of Picard Categories}\label{S:Picard Bundles}
In this section we will describe Picard bundles over a base space $M$.  These are meant to be presentations for bundles of Picard stacks, and thus are only ``bundles'' in the weak sense that they are bundles of short complexes (or Picard groupoids) defined locally on open subsets $U_\alpha\subset M$, and are glued together by chain automorphisms on intersections $U_\alpha\cap U_\beta$, and these gluings agree on triple intersections only up to a coherent set of homotopies.

For the base of our bundles we fix a topological space $M$ and a good cover $\M_0=\coprod_\alpha U_\alpha\to M$.  Associated to this cover is the \v{C}ech groupoid
\[ \M=(\M_1\arrows \M_0)=(\coprod U_\alpha\cap U_\beta\arrows \coprod U_\alpha). \]

Now a locally trivial bundle of groups with fiber $B$ may be presented via a \v{C}ech 1-cocycle, $\M_1\to \Aut(B)$, which is the same as a functor $\M\to (\Aut(B)\ltimes *)$.   A Picard bundle will be the result of generalizing from $B$ to $[A\to B]$, and upgrading the functor to a weak 2-functor $\M\dashto \AUT(A\to B)$.  We will describe this in detail now, using the short complex perspective.
\begin{definition}
Given a short complex $[A\tof{d} B]$, write $\Aut^0(A\to B)$ for the group of invertible degree zero chain morphisms from $A\to B$ to itself.  We define the 2-group $\AUT(A\to B)$ to be the groupoid whose objects are $\Aut^0(A\to B)$, and whose arrows are chain homotopies between these automorphisms.
\end{definition}

\begin{definition} A nonabelian cocycle $(\eta,\tau):\M\dashto \AUT(A\tof{d} B)$ is a pair of maps
\[ \M_1\Tof{\tau=(\tau_A,\tau_B)} \Aut^0(A\to B)\text{ and } \M_2\Tof{\eta}\Hom(B,A) \]
satisfying the cocycle conditions:
\begin{enumerate}
\item $\tau_{m_1}\circ\tau_{m_2}= \tau_{m_1m_2}+[d,\eta_{m_1,m_2}]$,
\item $\tau_{m_1}\circ\eta_{m_2,m_3}+\eta_{m_1,m_2m_3}= \eta_{m_1m_2,m_3}+\eta_{m_1,m_2}\circ\tau_{m_3}$,
\end{enumerate}
The maps $(\eta,\tau)$ are required to be continuous in the sense that $\M_1\times[A\to B]\tof{\tau} [A\to B]$ and $\M_2\times B\tof{\eta} A$ are continuous.
\end{definition}

\begin{definition} Let $(\eta,\tau):\M\dashto \AUT(A\to B)$ be a nonabelian cocycle.  The associated \textbf{Picard bundle}, denoted $[A\to B]_\M^{(\eta,\tau)}$ or $[A\to B]_\M^{(\eta,(\tau_A,\tau_B))}$ , refers to the (trivial) bundle of short complexes $\M_0\times (A\to B)$ together with the gluing-up-to-homotopy data $(\eta,\tau)$.

A morphism between two Picard bundles, denoted
\[ [A\to B]_\M^{(\eta,\tau)}\Tof{(\kappa, \phi)}[A'\to B']_\M^{(\eta',\tau')}, \]
is a pair of continuous maps
\[ \M_0\Tof{\phi}\Ch^0(A\to B,A'\to B')\qquad\M_1\Tof{\kappa} \Hom(B,A') \]
satisfying:
\begin{enumerate}
\item $\tau'_{m}\circ\phi_{sm}=\phi_{rm}\circ\tau_m+[d,\kappa_m]$,
\item $\tau'_{m_1}\circ\kappa_{m_2}+\kappa_{m_1}\circ\tau_{m_2}+\phi_{rm}\circ\eta_{m_1,m_2}=\eta'_{m_1,m_2}\circ\phi_{sm_2}+\kappa_{m_1m_2}$.
\end{enumerate}
A morphism is called \textbf{strict} if $\kappa=0$.  Note that for a strict morphism, Conditions (1)-(2) reduce to $\tau'\circ\phi=\phi\circ\tau$ and $\eta'\circ\phi=\phi\circ\eta$.
\end{definition}

Now we will show that the data of a Picard bundle induces a Picard groupoid over $\M$, and thus gives a presentation for a bundle of Picard stacks:
\begin{proposition}\label{P:Groupoid from Picard bundle}
\begin{enumerate}
\item Let $[A\tof{d}B]_\M^{(\eta,\tau)}$ be a Picard bundle.  Then the following structure defines a groupoid $\M^{(\eta,\tau)}$:
\[ \M^{(\eta,\tau)}:=(A\times B\times \M_1\arrows B\times \M_0) \]
\[ s(a,b,m):=(\tau_m^{-1}(b),sm),\quad r(a,b,m):=(da+b,rm) \]
\[ (a_1,\tau_{m_1}(da_2+\tau_{m_2}(b)),m_1)\circ (a_2,\tau_{m_2}(b),m_2):= (a_1+\tau_{m_1}(a_2)+\eta_{m_1,m_2}(b),\tau_{m_1m_2}(b),m_1m_2) \]
Furthermore, there is an obvious functor $\M^{(\eta,\tau)}\to\M$, and the group law
\[ \M^{(\eta,\tau)}\times_{\M}\M^{(\eta,\tau)}\Tof{+}\M^{(\eta,\tau)}\quad (a_1,b_1,m)+(a_2,b_2,m):=(a_1+a_2,b_1+b_2,m) \]
is functorial.
\item Let $[A\to B]_\M^{(\eta,\tau)}\Tof{(\kappa, \phi)}[A'\to B']_\M^{(\eta',\tau')}$ be a morphism of Picard bundles.  Then
\[ A\times B\times \M_1\to A'\times B'\times\M_1 \]
\[ (a,\tau_m(b),m)\longmapsto (\phi_{rm}(a)-\kappa_m(b),\phi_{rm}(\tau_m(b))+d\kappa_m(b),m) \]
is a functor $\M^{(\eta,\tau)}\to\M^{(\eta',\tau')}$ which intertwines the respective group laws.
\end{enumerate}
\end{proposition}
\begin{proof} This is just a matter of verifying associativity of the groupoid composition and functoriality of the addition, which is tedious but straightforward.
\end{proof}
\begin{remark} One can recover the cocycle $(\eta,\tau)$, and thus the Picard bundle $[A\tof{d}B]_\M^{(\eta,\tau)}$, from the groupoid structure maps of $\M^{(\eta,\tau)}$. Thus one could skip short complexes altogether and work from the start with groupoids of this form.  However, at least for us, the short complex perspective seems much cleaner.
\end{remark}
\begin{definition}  The \textbf{Pontryagin dual} of a Picard bundle $[A\to B]_{\M}^{(\eta,\tau)}$, is the Picard bundle $[\wh{B}\to\wh{A}]_{\M^{op}}^{(\wh{\eta},\wh{\tau})}$. \end{definition}
To verify that this definition makes sense, the reader should check that $(\wh{\eta},\wh{\tau})$ does indeed satisfy the conditions of a non-abelian cocycle, albeit on the opposite category $\M^{op}$. The opposite arises because Pontryagin duality reverses arrows, so that for instance
\[ \tau_{m_1}\circ\tau_{m_2}\rightsquigarrow \wh{\tau}_{m_2}\circ\wh{\tau}_{m_1}. \]

Let us now give, as examples, the types of Picard bundles that will be encountered in T-duality.  We use the notation $n's\in\Zb,\ z's\in U(1),\ t's\in T,\ v's\in V$, and write the group laws of $T$ and $U(1)$ multiplicatively, and all others additively:
\begin{example}\label{E:Picard Bundle Examples}\ \newline
\begin{enumerate}
\item A Picard bundle of the form $[0\to B]_\M^{(\eta,\tau)}$ is a presentation of a bundle of groups over $M$.  Indeed, necessarily $\eta=0$, so $\tau$ satisfies $\tau_{m_1}\tau_{m_2}=\tau_{m_1m_2}$ and is thus an equivalence relation on $\M_0\times B$.  The resulting quotient space $B_\tau:=\M_0\times B/\sim\tau$ is a locally trivial bundle of groups on $M$.
\item A sub-example of $(1)$ is the following: Given a principal torus bundle $P\to M$, the disjoint union $\coprod_{n\in \Zb}P^n$ of the $n$-th powers of $P$ is a locally trivial bundle of groups with fiber $\Zb\times T$. A choice of transition functions $\rho:\M_1\to T$ defining $P$ provides a Picard bundle presentation of $\coprod_{n\in \Zb}P^n$ as
    \[ \Ps= [0\to \Zb\times T]_\M^{(0,\tau)},\qquad  \tau_m(n,t):=(n,\rho^n_mt). \]
\item\label{Ex:Powers of Affine} A slight generalization of the previous example is the disjoint union of the $n$-th powers of an affine torus bundle.  The semi-direct product decomposition $\Aff(T)=T\rtimes\Aut(T)$ decomposes a 1-cocycle with values in $\Aff(T)$ as
    \[ \M_1\to T\times\Aut(T),\qquad m\mapsto(\rho_m,\tau^0_m). \]
    The cocycle condition implies that $\tau^0_{m_1}\tau^0_{m_2}=\tau^0_{m_1m_2}$ and also $\rho^n_{m_1}\tau^0_{m_1}(\rho^n_{m_2})=\rho^n_{m_1m_2}$ for all $n\in \Zb$.  Thus
    \[ \tau:\M_1\to\Aut(\Zb\times T),\qquad \tau_m(n,t):=(n,\rho^n_m\tau^0_m(t)) \]
    defines a Picard bundle $[0\to \Zb\times T]_\M^{(0,\tau)}$ which may be viewed as the union of the $n$-powers of the affine torus bundle.  Note that when $\tau^0=\id$, this reduces to the previous example.
\item Given a $U(1)$-gerbe $\Gc$ on a space $M$, a choice of \v{C}ech 2-cocycle $\sigma:\M_2\to U(1)$ representing its Dixmier-Douady class $[\sigma]\in H^2(M,U(1))$ gives rise to the following Picard bundle:
    \[  [U(1)\tof{0}\Zb]_\M^{(\eta,\id)},\qquad \eta_{m_1,m_2}(n):=\sigma_{m_1,m_2}^n, \]
     which is a presentation of the union of the n-th powers of $\Gc$ (by n-th power we mean the gerbe whose D-D class is $\sigma^n$).  In fact, according to Proposition \eqref{P:Groupoid from Picard bundle}, we have an associated groupoid $\M^{(\eta,\tau)}$ whose arrow space is $U(1)\times\Zb\times\M_1$, and for each $n\in \Zb$, the sub-groupoid $U(1)\times\{n\}\times\M_1\subset U(1)\times\Zb\times\M_1$ is the standard groupoid presentation of the n-th power of $\Gc$.
\item Combining the two previous examples, consider the morphism
\[ \Gs:=[U(1)\tof{0}\Zb\times T]_\M^{(\eta,(\id_{U(1)},\tau))}\To \Ps=[0\to \Zb\times T]_\M^{(0,\tau)} \]
 induced by $U(1)\mapsto 0$.  This makes $\M^{(\eta,(\id_{U(1)},\tau))}$ a (groupoid presentation of) a $U(1)$-gerbe over $\M^{(0,\tau)}$.  Furthermore, $\eta$ may be decomposed as $\eta_{m_1,m_2}(n,t)=\sigma_{m_1,m_2}^n\eta^0_{m_1,m_2}(t)$, and for each $n$, $\eta(n,-):\M_2\times T\to U(1)$ is a groupoid 2-cocycle presenting a $U(1)$-gerbe on the affine torus bundle $P^n=\M_0\times T/\sim \tau(n,-)$.

 In particular, setting $n=0$ we have
 \[ \Gs^0=[U(1)\to T]^{(\eta^0,(id_{U(1)},\tau^0))}\To \Ps^0:=[0\to T]^{(0,\tau^0)} \]
 which is a Picard bundle presentation of a gerbe over the $0$-th power of $P$. We will refer to $\Gs^0\subset\Gs$ as the \textbf{degree-0} piece of $\Gs$.
\end{enumerate}
\end{example}
\begin{remark} With regards to the last example, we should remark that not every $U(1)$-gerbe on an affine or principal torus bundle $P\to M$ extends to a $U(1)$-gerbe on $\coprod P^n$.  In fact, we will see that this property characterizes those gerbes which admit T-duals.
\end{remark}

\section{Construction of T-duals}\label{S:Construction of T-duals}
In this section we construct the $T$-dual of a $U(1)$-gerbe on an affine torus bundle $P$.  We will use the notation developed in Example \eqref{E:Picard Bundle Examples}.  The construction relies on two assumptions:
\begin{enumerate}
\item The gerbe on $P$ extends to a gerbe on the Picard bundle \linebreak$\Ps:=[0\to \Zb\times T]^{(0,\tau)}$ representing the union of the n-th powers of $P$.
\item The canonical quotient morphism $[0\to \Zb\times T]^{(0,\tau)}\to [V\tof{q}\Zb\times T]^{(0,(\tau^0,\tau))}$ lifts to the gerbe.  Here
\[ (\tau^0,\tau):\M_1\to \Aut^0(V\tof{q}\Zb\times T),\quad (v,n,t)\longmapsto (\tau^0(v),n,\tau(n,t)),\]
which makes sense as long as $\tau^0:\M_1\to\Aut(T)$ acts on $V$ by the inclusion $\Aut(T)\subset\Aut(V)$.
\end{enumerate}
We explain in Remark \eqref{R:Assumption 1 implies 2} that the first assumption implies the second.  We also claim that the first assumption automatically holds for any $U(1)$-gerbe on a principal bundle which admits a (commutative) T-dual in the sense of \cite{MatRos} or \cite{BunSch}, though we will not prove it in this article.

Now let us proceed with the construction.  It will involve a quotient by the bundle of groups $V_{\tau^0}:=\M_0\times V/\sim{\tau^0}$, then a Pontryagin dualization, and then an optional step which places the dual in a nice form.

Recall the notation that $\Ps:=[0\to \Zb\times T]^{(0,\tau)}$ is the union of the $n$-th powers of an affine torus bundle, $\Gs:=[U(1)\to \Zb\times T]^{(\eta,(id_{U(1)},\tau)}$ is a Picard bundle representing a $U(1)$-gerbe over $\Ps$, and $\Gs_0=[U(1)\to \Zb\times T]^{(\eta^0,(id_{U(1)},\tau^0)))}$ is the degree-0 piece of $\Gs$.  There is a canonical quotient morphism
\[  [0\to \Zb\times T]^{(0,\tau)}\To [V\tof{q}\Zb\times T]^{(0,(\tau^0,\tau))}=:\Ps//V_{\tau^0}. \]
Lifting this quotient to $\Gs$ amounts to finding a non-abelian cocycle $(\wt{\eta},(\alpha,\tau))$ making the following diagram commute:
\begin{equation}\label{E:V-quotient}
\xymatrix{
\Gs=[U(1)\to \Zb\times T]^{(\eta,(\id_{U(1)},\tau))}_\M \ar[rr]\ar[d] && \Gs//V_{\tau^0}:=[U(1)\times V\tof{0\times q} \Zb\times T]^{(\wt{\eta},(\alpha,\tau))}_\M\ar[d] \\
\Ps=[0\to \Zb\times T]^{(0,\tau)}           \ar[rr]       && \Ps// V_{\tau^0}:=[V\tof{q}\Zb\times T]^{(0,(\tau^0,\tau))}.
} \end{equation}

In the upper left corner, we have $\eta:\M_2\to \Hom(\Zb\times T,U(1))$.  We claim that in the upper right corner $\wt{\eta}$ must be the same map, albeit viewed as taking values in $\Hom(\Zb\times T,U(1)\times V)$.  Indeed, the $U(1)$-value of $\wt{\eta}$ in the upper right must agree with the $\eta$ in the upper left, and the $V$-value of $\wt{\eta}$ must be $0$ to agree with the bottom right.  From now on we will write $\wt{\eta}=\eta$.

Thus the only new map is $\alpha$, which describes the lifted $V$-action.  By commutativity of the diagram, $\alpha:\M_1\to \Aut(U(1)\times V)$ must be of the form
\[ \alpha_m(z,v)=(\psi_m(v)z,\tau^0_m(v)) \]
for some $\psi:\M_1:\to \Hom(V,U(1))=\wh{V}$, and the reader can check that $(\eta,(\alpha,\tau))$ is a non-abelian cocycle if and only if
\begin{equation}\label{E:bounding cochain} (\psi_{m_1}\circ \tau^0_{m_2})(v)\psi_{m_2}(v)=\psi_{m_1m_2}(v)\eta^0_{m_1,m_2}(qv). \end{equation}

\begin{remark}\label{R:Assumption 1 implies 2}[Asumption (1) implies Assumption (2) and Uniqueness] The cocycle condition satisfied by $q^*\eta^0:\M_2\to \wh{V}$ makes it a \v{C}ech representative for a degree 2 cohomology class with values in the sheaf of sections of the bundle $\wh{V_{\tau^0}}$.  But the existence of (continuous) partitions of unity imply that $H^{k>0}(M,\wh{V_{\tau^0}})=0$, so that $q^*\eta^0$ must be a boundary.  Thus the bounding cochain $\psi$ of Equation \eqref{E:bounding cochain} is certain to exist.  Furthermore, any two such cochains differ by a 1-cocycle, which must itself be a boundary.  This implies that the $V_{\tau^0}$-quotient is not only certain to exist, but is unique up to equivalence.
Of course in the algebro-geometric setting these cohomology groups do not vanish, so Assumption (2) is necessary.
\end{remark}

\begin{remark} The map $\psi$ induces a morphism of Picard bundles
\[ [0\to V]_\M^{(0,\tau^0)}\Tof{(\phi,\kappa)}[U(1)\to \Zb\times T]^{(\eta,(\id_{U(1)},\tau))}_\M \]
where $\phi= q:V\to T$ and $\kappa:=\psi$.  Thus the right side of \eqref{E:V-quotient} is literally the quotient $\Gs//V_{\tau^0}$ by a homomorphism.
\end{remark}

Now we proceed to Step (2), which is Pontryagin dualization of $\Gs//V_{\tau^0}$.  Keeping in mind that $\wh{\Zb}=U(1)$, this is:
\[
\Pont([U(1)\times V\tof{0\times q} \Zb\times T]^{(\eta,(\alpha,\tau))}_\M \ = \ [U(1)\times \wh{T}\tof{0\times q^*} \Zb\times \wh{V}]_{\M^{op}}^{(\wh{\eta},(\wh{\tau},\wh{\alpha}))}
\]
The inclusion $\Gs^0//V_{\tau^0}\into\Gs//V_{\tau^0}$ corresponds fiberwise to $0\into \Zb$, which dualizes to $U(1)\to 0$.  Thus we obtain a morphism of Picard bundles:
\begin{equation}\label{E:T-dual-first-form} \xymatrix{
\Pont(\Gs//V_{\tau^0}):= & [U(1)\times \wh{T}\tof{0\times \wh{q}} \Zb\times \wh{V}]_{\M^{op}}^{(\wh{\eta},(\wh{\tau},\wh{\alpha}))}\ar[d]\\ \Pont(\Gs_0// V_{\tau^0}):= & [\wh{T}\tof{\wh{q}} \Zb\times \wh{V}]_{\M^{op}}^{(\wh{\eta^0},(\wh{\tau^0},\wh{\alpha}))}}.
\end{equation}
We claim that the object downstairs is a Picard bundle presentation of the union of $n$-th powers of an affine torus bundle, which is (by definition) the $T$-dual torus bundle.  Furthermore, the object upstairs is a $U(1)$-gerbe over the base, which is by definition the $T$-dual gerbe.  These claims are not obvious, but at least fiberwise they seem plausible if we keep in mind that $[\wh{T}\to \wh{V}]$ is a presentation of the T-dual torus $T^\vee=\wh{V}/\wh{T}$.

Now we proceed to Step (3), which we state as our main theorem.  It amounts to a recasting of the objects in Equation \eqref{E:T-dual-first-form} into a nice form that makes the claims of the previous paragraph obvious, and makes $T$-dualization obviously symmetric.

\begin{theorem}[Main Theorem] Suppose that $\Gs:=[U(1)\to \Zb\times T]_\M^{(\eta,(\id_{U(1)},\tau))}$  is a Picard bundle presentation of a gerbe over an affine torus bundle $\Ps:=[0\to \Zb\times T]_\M^{(0,\tau)}$.  Then $\Gs$ admits a T-dualization, which is unique up to equivalence.  The T-dual may be presented as:
\begin{equation}\label{E:T-dual-final-form} \xymatrix{
\Gs^\vee:= & [U(1)\to \Zb\times T^\vee]_{\M^{op}}^{(\eta^\vee,(\id_{U(1)},\tau^\vee))} \ar[d]\\ \Ps^\vee:= & [0 \to\Zb\times T^\vee]_{\M^{op}}^{(0,\tau^\vee)}}
\end{equation}
where $\tau^\vee_m(n,t^\vee):=(n,\psi^n_m\wh{\tau^0}_m(t^\vee))$ and $\eta^\vee_{m_1,m_2}(n,t^\vee):=\sigma^{\vee n}_{m_1,m_2}\eta^{\vee 0}(t^\vee)$.  Here
\[ \sigma^\vee_{m_1,m_2}:=\sigma_{m_1,m_2}\langle\wt{\rho}_{m_1}^{-1},\psi_{m_2}\rangle,\qquad \eta^{\vee 0}=(\wt{\rho}^{-1}_{m_1}\circ\tau^0_{m_2})\wt{\rho}_{m_2}^{-1} \wt{\rho}_{m_1m_2},\]
and $\tau=(\tau^0,\rho)\in\Aut(T)\ltimes T$ defines the original affine bundle, $\wt{\rho}:\M_1\to V$ is any lift of $\rho$, and $\psi$ is a bounding cocycle for $\eta^0$ defining the $V_{\tau^0}$-quotient.
\end{theorem}
\begin{proof}  We will just write down a morphism of Picard bundles\linebreak $\Pont(\Gs//V_{\tau^0})\to\Gs^\vee$ which is an equivalence.  Thus define
\[
[U(1)\times \wh{T}\tof{0\times \wh{q}} \Zb\times \wh{V}]_{\M^{op}}^{(\wh{\eta},(\wh{\tau},\wh{\alpha}))}\Tof{(\phi,\kappa)} [U(1)\to \Zb\times T^\vee]_{\M^{op}}^{(\eta^\vee,(id_{U(1)},\tau^\vee))}
\]
by letting $\phi$ be the constant quotient map $[U(1)\times \wh{T}\to \Zb\times \wh{V}]\Tof{\phi} [U(1)\to \Zb\times T^\vee]$, and setting
\[ \kappa:\M_1\to \Hom(\Zb\times \wh{V},U(1)),\qquad \kappa_m(n,\wh{v}):=\langle\wt{\rho}^{-1}_m,\wh{v}\rangle. \]
Then modding out $U(1)$ sends $\kappa$ to the zero map, thus we obtain from $(\phi,\kappa)$ a strict equivalence $\Pont(\Gs^0//V_{\tau^0})\to\Ps^\vee$:
\[
[\wh{T}\tof{0\times \wh{q}} \Zb\times \wh{V}]_{\M^{op}}^{(\wh{\eta^0},(\wh{\tau^0},\wh{\alpha}))}\Tof{(\phi,0)} [0\to\Zb\times T^\vee]_{\M^{op}}^{(0,\tau^\vee))}.
\]
\end{proof}
One can also easily verify the $\Pont(\Ps//V_{\tau^0})$ becomes the degree-0 part of $\Gs^\vee$. This fact about T-dualization, that the roles of $\Ps$ and $\Gs^0$ are interchanged, is quite important.  It means in particular that taking the union of powers of the affine bundle $P$ is important, because all of those powers contribute to the gerbe on the T-dual side.

\section{Conclusion}\label{S:Relation}
In summary, we have described a concrete method of T-dualizing a $U(1)$-gerbe $\Gc\to P$ over an affine torus bundle, which depends on embedding it into a $U(1)$-gerbe $\Gs\to \Ps$ over the union of all powers of the affine bundle.  Our method produces a unique T-dual once such an embedding is chosen (though different embeddings may result in non-equivalent T-duals, which is expected, because T-duality is not typically 1-to-1), thus if defined from the beginning as an operation on group stacks, T-dualization may even be viewed as a functor.

We also saw that T-dualization applies separately to the bundle $\Ps$, the gerbe $\Gs$, and also the degree zero component $\Gs_0$ of $\Gs$.  This uncovers a core feature of T-duality, which is that it interchanges the roles of $\Ps$ and $\Gs_0$--- $\Ps$ becomes the degree zero component of the dual gerbe, and $\Gs^0$ becomes the dual torus bundle.

Finally, we should indicate how this relates to the topological, algebro-geometric, and differential geometric models of T-duality.  

Our setup is closest to the topological version of \cite{MatRos}, in fact if $\Gs\to\Ps$ is a Picard bundle presentation of a $U(1)$-gerbe, then we may take the associated groupoid as in Proposition \eqref{P:Groupoid from Picard bundle}, and then form the $C^*$-algebra of this groupoid.  Then modding out $V$ becomes the operation of taking the crossed product $C^*$-algebra for the $V$-action.  Pontryagin duality becomes Fourier transform, which is an isomorphism and may be ignored in the $C^*$-algebra picture.

To make contact with the algebro-geometric picture, we should produce a Poincar\'e line bundle which induces the equivalence of derived categories as in \cite{DonPan}.  We will just describe the case of a single torus (which was pointed out to us by Martin Olsson).  In this case
there is a canonical 1-dimensional complex line bundle on $[U(1)\to 0]$--- it is simply $\Cb$ with the usual action of $U(1)$.  There is also a canonical morphism of stacks:
\[ T\times T^\vee\Tof{q} [T//V]\times T^\vee\Tof{\ev}[*//U(1)] \]
where ``$\ev$'' is the evaluation map on the product of $[T//V]\simeq[*//\Lambda]$ with its Pontryagin dual $T^\vee=\Pont([*//\Lambda])$.  The reader can verify that
\[ \Lc:=q^*\ev^*(\Cb)\in \textrm{Vect}(T\times T^\vee) \]
is in fact the Poincar\'e line bundle.

In the differential geometric picture of \cite{CavGua} (again we will consider only a single torus), one produces an isomorphism between the translation invariant sections of the generalized tangent bundle of $T$ and those of $T^\vee$, and from that produces an isomorphism between the generalized geometrical structures on $T$ and those on $T^\vee$.  The isomorphism requires a ``dualizing'' 2-form  $F\in\Omega^2(T\times T^\vee)$, which in this case may be taken to be the curvature of the line bundle $\Lc$.

It should be noted however, that in the setup of \cite{CavGua}, $F$ guarantees the existence of an isomorphism but does not produce a canonical one.  Thus an interesting question for further study is whether T-dualization might be used to transport geometric structures across T-duality in a direct and canonical way.  To do this one would need to understand generalized geometrical structures on a Picard stack, which to our knowledge have not been studied.

\bibliographystyle{amsalpha}

\end{document}